\documentclass[a4paper]{article}

\usepackage[utf8]{inputenc}
\usepackage[T1]{fontenc}

\usepackage[left=2cm,right=2cm,top=2cm,bottom=2cm]{geometry}
\usepackage[english]{babel}
\usepackage{graphicx}
\usepackage[table]{xcolor}            
\usepackage[labelfont=bf,font=small]{caption}
\usepackage{subcaption}
\usepackage[shortlabels]{enumitem}
\usepackage{hyperref}
\usepackage{amsmath,amssymb,amsthm}
\usepackage{stackengine}
\stackMath

\newtheorem{definition}{Definition}
\newtheorem{lemma}{Lemma}
\newtheorem{theorem}{Theorem}

\newtheorem{corollary}{Corollary}

\newcommand{\blue}[1]{{\color{blue} {#1}}}
\newcommand{\red}[1]{{\color{red} {#1}}}

\long\def\ignore#1{}

\renewcommand{\H}{\mathcal{H}}
\newcommand{\Leq}{\;\dot{=}\;}

\newcommand{\types}{\mathcal{T}}
\newcommand{\Var}{\mathcal{V}}

\newcommand{\syntrue}{\top}
\newcommand{\synfalse}{\bot}

\newcommand{\SigmaOne}{\Sigma^{\dot{=}}}
\newcommand{\SigmaEq}{\Sigma^{=}}

\title{On Reductions of Hintikka Sets for Higher-Order Logic}
\author{Alexander Steen$^1$, Christoph Benzmüller$^2$\\
        \small
        $^1$ University of Luxembourg, FSTM, \href{mailto:alexander.steen@uni.lu}{\texttt{alexander.steen@uni.lu}}\\
        \small
        $^2$ Freie Universität Berlin, FB Mathematik und Informatik,\href{mailto:c.benzmueller@fu-berlin.de}{\texttt{c.benzmueller@fu-berlin.de}}
       }
\date{\today}

\begin{document}
\maketitle

\begin{abstract}
  Steen's (2018) Hintikka set properties for Church's type theory
  based on primitive equality are reduced to the Hintikka set
  properties of Brown (2007). Using this reduction, a model existence theorem
  for Steen's properties is derived.
\end{abstract}

%%%%%%%%%%%%%%%%%%%%%%%%%%%%%%%%%%%%%%%%%%%%%%%%%%%%%%%%%%%%%%%%%%%%%%%%%%%%%%%%%%%%%%%%
%%%%%%%%%%%%%%%%%%%%%%%%%%%%%%%%%%%%%%%%%%%%%%%%%%%%%%%%%%%%%%%%%%%%%%%%%%%%%%%%%%%%%%%%
%%%%%%%%%%%%%%%%%%%%%%%%%%%%%%%%%%%%%%%%%%%%%%%%%%%%%%%%%%%%%%%%%%%%%%%%%%%%%%%%%%%%%%%%
\section{Introduction and Preliminaries}
Abstract consistency properties and Hintikka sets play an important
role in the study (e.g.,~of Henkin-complete\-ness) of proof calculi
for Church's type theory~\cite{church40,J43}, aka. classical
higher-order logic (HOL). Technically quite different definitions of
these terms have been used in the literature, since they depend on the
primitive logical connectives assumed in each case. The definitions of
Benzmüller, Brown and Kohlhase \cite{J6,R37}, for example, are based
on negation, disjunction and universal quantification, while Steen
\cite{DBLP:phd/dnb/Steen18}, in the tradition of
Andrews~\cite{andrews2002introduction}, works with primitive equality
only. Despite their conceptual relationship, important semantical
corollaries that are implied by these syntax related Hintikka properties (such
as model existence theorems) can hence not be directly transferred
between formalisms. 

Brown~\cite{brown07,brown04} presents generalised abstract consistency
properties in which the primitive logical connectives can vary. In
this paper we show that the properties of Steen can be reduced to
those of Brown.\footnote{In an earlier reduction attempt we tried to
  reduce Steen's \cite{DBLP:phd/dnb/Steen18} abstract consistency
  properties to those of Benzmüller, Brown and Kohlhase
  \cite{J6}. This attempt had gaps that could not easily be closed (as was pointed out by an unknown
  reviewer). The mapping that we applied in this initial reduction attempt between the
  two formalisms replaced primitive equations (in Steen's Hintikka
  sets) by Leibniz equations (to obtain Hintikka sets in the style of
  Benzmüller, Brown and Kohlhase). It thereby introduced additional
  universally quantified formulas which in turn triggered the
  applicability of abstract consistency conditions whose validity
  could not be ensured by referring to those of Steen. In this work we
  therefore  instead map to the
  structurally better suited, generalised conditions
  of Brown~\cite{brown07,brown04}, which enables us to circumvent the
  problems induced by our previous reduction attempt to the work of Benzmüller, Brown and Kohlhase.\label{footnote:1}}
Theorem~\ref{thm:reduction} that we establish in this paper paves way
for the convenient reuse of (e.g.,~model existence) results from the
work of Brown~\cite{brown07,brown04}  in the context of Steen's setting.
%by showing that Steen's Hintikka
%set properties can be reduced to those as studied by Benzmüller et al. 
%In more general terms we illustrate how technical dependencies on particular primitive logical connectives can
%be overcome by providing respective reductions between technically different definitions of Hintikka sets.

\paragraph{Paper structure.} In the remainder of this introduction we recapitulate
some relevant (syntactic) notions on HOL, mainly to clarify our
notation; for further details on HOL as relevant for this
paper.
%publications of Steen and Benzmüller et
%al.~\cite{J6,R37,DBLP:phd/dnb/Steen18}. 
\S\ref{sec:chad} we present the Hintikka properties as used by
Brown, and in \S\ref{sec:steen} we give the related
properties as used by Steen. In \S\ref{sec:steen} we then show various
lemmas that are implied in Steen's setting, and it are (some of) those lemmas which prepare the main
reduction result of this paper (Theorem \ref{thm:reduction}), which is given
in \S\ref{sec:reduction}. A model existence theorem for Hintikka sets as defined by
Steen is then derived in \S\ref{sec:bridge}.

\paragraph{Equality conventions.}
Different notions of equality will be used in the following: If a concept is defined
(as an abbreviation), the symbol $:=$ is used.
Primitive equality, written $=^\tau$, refers to a
logical constant symbol
from the HOL language such that $s_\tau =^\tau t_\tau$ is a term of the logic
(assuming $s_\tau$ and $t_\tau$ are, where $\tau$ is a type
annotation), cf.\ details further below.
 Leibniz-equality, written $\dot{=}$, is a defined term; usually it stands for $\lambda X_\tau.\, \lambda
Y_\tau.\, \forall P_{o\tau}.\, (P\;X) \Rightarrow (P\;Y)$, where
$\Rightarrow$ is a (primitive) logical connective.
Meta equality $\equiv$ denotes set-theoretic identity between objects.
Finally, $\equiv_\star$, for $\star \subseteq \{\beta,\eta\}$ is used for syntactic
equality modulo $\beta-$, $\eta-$ and $\beta\eta$-conversion,
respectively (as in the related work, $\alpha$-conversion
is taken as implicit).

\paragraph{Syntax of HOL.}
The set $\types$ of simple types is freely generated from the base types $o$ and $\iota$
by juxtaposition. The types $o$ and $\iota$ represent
the type of Booleans and individuals, respectively. A type $\nu\tau$ represents the type of a
total function from objects of type $\tau$ to objects of type $\nu$.

Let $\Sigma_\tau$ be a set of constant symbols of type $\tau \in \types$
and let $\Sigma := \bigcup_{\tau \in \types} \Sigma_\tau$ be the union of all typed symbols, called a \emph{signature}.
Let further $\Var$ denote a set of (typed) variable symbols. From these the terms
of HOL are constructed by the following abstract syntax ($\tau,\nu \in \types$):
\begin{equation*}
s,t ::= c_\tau \in \Sigma \; | \; X_\tau \in \Var \; | \; \left(\lambda X_\tau.\, s_\nu\right)_{\nu\tau}
  \; | \; \left(s_{\nu\tau} \; t_\tau\right)_\nu
\end{equation*}
The terms are called \textit{constants}, \textit{variables}, \textit{abstractions}
and \textit{applications}, respectively.  \index{Application} \index{Abstraction} \index{Constant} \index{Variable}
The set of all terms of type $\tau$ over a signature $\Sigma$
is denoted $\Lambda_\tau(\Sigma)$, and $\Lambda_\tau^c(\Sigma)$ is
used for closed terms, respectively. The notion of free and bound variables are defined as usual, and a term $t$ is
called \emph{closed} if $t$ does not contain any free variables.
It is assumed that the set $\Var$ contains countably infinitely many variable symbols for each
type $\tau \in \types$.

The type of a term is written as subscript but may be dropped by convention if
clear from the context (or if not important).
Also, parentheses are omitted whenever possible, and application is assumed to be
left-associative. Furthermore, the scope of an $\lambda$-abstraction's body reaches as far 
to the right as is consistent with the remaining brackets.
Nested applications $s\;t^1\;\ldots\;t^n$ may also be written in vector notation
$s\;\overline{t^n}$.

Variants of HOL often differ regarding the choice of the primitive
logical connectives in the signature
$\Sigma$. 
In any case, $s \neq t$ is used in the remainder as an abbreviation for $\neg(s = t)$.
Also, for simplicity, binary logical connectives
may be written in infix notation; e.g., the term $p_o \lor q_o$ formally represents 
the application $\left(\lor_{ooo} \; p_o \; q_o\right)$. Furthermore, binder notation is used for universal and existential
quantification: The term $\forall X_\tau.\, s_o$ is used as a
short-hand for $\Pi^\tau\left(\lambda X_\tau.\, s_o \right)$, where
$\Pi^\tau$ is a constant symbol.
%and analogously for existential quantification.
Finally, Leibniz-equality, denoted $\dot{=}$, is defined as
$\dot{=} := \lambda X_\tau.\, \lambda Y_\tau.\, \forall P_{o\tau}.\, (P\;X) \Rightarrow (P\;Y)$.
A \emph{$\Sigma$-formula} $s_o$ is a term from $s_o \in \Lambda_o(\Sigma)$ of type $o$ and a \emph{$\Sigma$-sentence}
if it is a closed $\Sigma$-formula. The reference to $\Sigma$ may be omitted if clear from the context.

In the following, variables are denoted by capital letters such as $X_\tau, Y_\tau, Z_\tau$, and, more specifically,
the variable symbols $P_o, Q_o$ and $F_{\nu\tau}, G_{\nu\tau}$ are used for predicate or Boolean variables
and variables of functional type, respectively. Analogously, lower case letters $s_\tau, t_\tau, u_\tau$ denote
general terms and $f_{\nu\tau}, g_{\nu\tau}$ are used for terms of functional type.

\paragraph{Semantics of HOL.}
The semantics of HOL, including the notions of $\Sigma$-models and $\Sigma$-Henkin models, 
is not discussed here; cf.~\cite{J6,DBLP:phd/dnb/Steen18,brown07,brown04} for details.

%%%%%%%%%%%%%%%%%%%%%%%%%%%%%%%%%%%%%%%%%%%%%%%%%%%%%%%%%%%%%%%%%%%%%%%%%%%%%%%%%%%%%%%%
%%%%%%%%%%%%%%%%%%%%%%%%%%%%%%%%%%%%%%%%%%%%%%%%%%%%%%%%%%%%%%%%%%%%%%%%%%%%%%%%%%%%%%%%
%%%%%%%%%%%%%%%%%%%%%%%%%%%%%%%%%%%%%%%%%%%%%%%%%%%%%%%%%%%%%%%%%%%%%%%%%%%%%%%%%%%%%%%%
\newcommand{\rneg}{\red{\neg}}
\newcommand{\req}{\mathbin{\red{=}}}
\newcommand{\rneq}{\mathbin{\red{\neq}}}
\newcommand{\rlor}{\mathbin{\red{\lor}}}
\newcommand{\rland}{\mathbin{\red{\land}}}
\newcommand{\rpi}{\red{\Pi}}
\newcommand{\rtop}{\red{\top}}
\newcommand{\rbot}{\red{\bot}}
\newcommand{\rleq}{{\red{\Leq}}}

\newcommand{\dneg}{\stackunder[-2pt]{\neg}{\cdot}}
\newcommand{\deq}{\mathbin{\stackunder[-1pt]{=}{\cdot}}}
\newcommand{\dlor}{\mathbin{\stackunder[0pt]{\lor}{\cdot}}}
\newcommand{\dland}{\mathbin{\stackunder[0pt]{\land}{\cdot}}}
\newcommand{\dforall}{\stackunder[0pt]{\forall}{\cdot}}
\newcommand{\dexists}{\stackunder[0pt]{\exists}{\cdot}}
\newcommand{\dtop}{\stackunder[0pt]{\top}{\cdot}}
\newcommand{\dbot}{\stackunder[0pt]{\bot}{\cdot}}

\newcommand{\rdneg}{\red{\dneg}}
\newcommand{\rdeq}{\mathbin{\red{\deq}}}
\newcommand{\rdlor}{\mathbin{\red{\dlor}}}
\newcommand{\rdland}{\mathbin{\red{\dland}}}
\newcommand{\rdforall}{\red{\dforall}}
\newcommand{\rdexists}{\red{\dexists}}
\newcommand{\rdtop}{\red{\dtop}}
\newcommand{\rdbot}{\red{\dbot}}

\section{Hintikka sets as defined by Brown~\cite{brown07} \label{sec:chad}}
In the formulation of HOL as employed by Brown~\cite{brown07},
the set of primitive logical connectives is not fixed and may be chosen
arbitrarily from the set $\{ \top_o, \bot_o, \neg_{oo}, \wedge_{ooo}, \vee_{ooo}, \supset_{ooo},
\equiv_{ooo}\} \cup \{\Pi_{o(o\tau)} | \tau\in \types\} \cup
\{\Sigma_{o(o\tau)} \mid \tau\in \types\} \cup \{=^\tau_{o\tau\tau} \mid \tau\in \types\}$. 
All remaining constant symbols from $\Sigma$ are called parameters.
%Depending on the choice of primitive logical connectives, the usual remaining connectives
%might be definable. 
In this presentation, the notation from Brown is adapted and we apply
the conventions from above.  For example, we denote general terms with
lower case symbols $s$ and $t$ instead of upper case letters as used
by Brown, and instead of $\mathit{wff}_\tau(\Sigma)$, which Brown uses to denote
the set of internal terms of type $\tau$, we use
$\Lambda_\tau(\Sigma)$.

Brown distinguishes between internal terms, the elements of
$\Lambda_\tau(\Sigma)$, and external propositions (meta-level propositions) in a set $prop(\Sigma)$, which are defined
as follows \cite[Def.~2.1.20]{brown07} (we use the color \red{red} to visually highlight Brown's
meta-level connectives and their associated Hintikka set properties in
the remainder of this paper):
\begin{itemize}
\item If $s\in\Lambda_o(\Sigma)$, then $s\in prop(\Sigma)$.
\item If $\alpha\in\types$ and $s,t \in  \Lambda_\alpha(\Sigma)$, then
  $[s\rdeq t]\in prop(\Sigma)$.
\item $\rdtop \in prop(\Sigma)$.
\item If $s\in prop(\Sigma)$, then $[\rdneg s]\in prop(\Sigma)$.
\item If $s, t \in prop(\Sigma)$, then $[s\rdlor t]\in prop(\Sigma)$.
\item If $s \in prop(\Sigma)$, then $[\rdforall X_\alpha s]\in prop(\Sigma)$.
\end{itemize}

Closed propositions $s\in prop(\Sigma)$ are also called sentences; we
then write $s\in sent(\Sigma)$.

Brown introduces the following Hintikka set properties for
external propositions.

\paragraph{Properties of extensional Hintikka sets $\H$~\cite[Def. 5.5.4]{brown07}:}
\mbox{}\\

$\red{\vec{\nabla}_c}:$ $s \notin \H$ or $\rdneg s \notin \H$.

$\red{\vec{\nabla}_{\beta\eta}}:$ If $s \in \H$, then $s^{\downarrow} \in \H$.

$\red{\vec{\nabla}_\bot}:$ $\rdneg \rdtop \notin \H$.

$\red{\vec{\nabla}_\neg}:$ If $\rdneg\rdneg s \in \H$, then $s \in H$.

$\red{\vec{\nabla}_\lor}:$ If $s \rdlor t \in \H$, then $s \in \H$ or $t \in \H$.

$\red{\vec{\nabla}_\land}:$ If $\rdneg(s \rdlor t) \in \H$,
  then $\rdneg s \in \H$ and $\rdneg t \in \H$.

$\red{\vec{\nabla}_\forall}:$ If $\rdforall X_\tau\, s \in \H$, then
  $[t/x]s\in \H$ for every closed term $t \in \Lambda^c_\tau(\Sigma)$.

$\red{\vec{\nabla}_\exists}:$ If $\rdneg(\rdforall X_\tau\, s) \in \H$,
then there is a parameter $p_\tau \in \Sigma_\tau$ such that
  $\rdneg([p/X]s) \in \H$.

$\red{\vec{\nabla}^\sharp}:$ If $s \in \H$, then $s^\sharp \in \H$. Also,
  if $\rdneg s \in \H$, then $\rdneg s^\sharp \in \H$.

$\red{\vec{\nabla}_{\mathit{m}}}:$ If $\rdneg\big(h\;\overline{s^n}\big) \in \H$
  and $\big(h\;\overline{t^n}\big) \in \H$, then there is an $i$ with $1 \leq i \leq n$
  such that $\rdneg\big(s^i \rdeq t^i\big) \in \H$.

$\red{\vec{\nabla}_{\mathit{dec}}}:$ If $p$ is a parameter and
  $\rdneg\big((p\;\overline{s^n}) \rdeq^\iota (p\;\overline{t^n})\big) \in \H$,
  then there is an $i$, $1 \leq i \leq n$, s.t.\
  $\rdneg\big(s^i \rdeq t^i\big) \in \H$.

$\red{\vec{\nabla}_{\mathfrak{b}}}:$ If $\rdneg(s \rdeq^o t) \in \H$, then 
  $\{s, \rdneg t\} \subseteq \H$ or $\{\rdneg s, t\} \subseteq \H$.

$\red{\vec{\nabla}_{\mathfrak{f}}}:$ If $\rdneg(f \rdeq^{\nu\tau} g) \in \H$,
  then there is a parameter $p_\tau \in \Sigma_\tau$
  such that $\rdneg(f\;p \rdeq^\nu g\; p) \in \H$.

$\red{\vec{\nabla}_{=}^o}:$ If $s \rdeq^o t \in \H$, then
  $\{s, t\} \subseteq \H$ or $\{\rdneg s, \rdneg t\} \subseteq \H$.

$\red{\vec{\nabla}_{=}^\rightarrow}:$ If $f \rdeq^{\nu\tau} g \in \H$, then
  $(f\;u \rdeq^\nu g \; u) \in \H$ for every closed
  term $u \in \Lambda^c_\tau(\Sigma)$.

$\red{\vec{\nabla}_{=}^r}:$ $\rdneg(s \rdeq^i s) \notin \H$.

$\red{\vec{\nabla}_{=}^u}:$ Suppose $(s \rdeq^i t) \in \H$ and $\rdneg(u \rdeq^i v) \in \H$.
  Then $\rdneg(s \rdeq^i u) \in \H$ or $\rdneg(t \rdeq^i v) \in \H$.
  Also, $\rdneg(s \rdeq^i v) \in \H$ or $\rdneg(t \rdeq^i u) \in \H$. \\

\noindent The collection of all sets of sentences satisfying all these properties
is called $\mathfrak{Hint}_{\beta\mathfrak{f}\mathfrak{b}}(\Sigma)$.

%%%%%%%%%%%%%%%%%%%%%%%%%%%%%%%%%%%%%%%%%%%%%%%%%%%%%%%%%%%%%%%%%%%%%%%%%%%%%%%%%%%%%%%%
%%%%%%%%%%%%%%%%%%%%%%%%%%%%%%%%%%%%%%%%%%%%%%%%%%%%%%%%%%%%%%%%%%%%%%%%%%%%%%%%%%%%%%%%
%%%%%%%%%%%%%%%%%%%%%%%%%%%%%%%%%%%%%%%%%%%%%%%%%%%%%%%%%%%%%%%%%%%%%%%%%%%%%%%%%%%%%%%%
%%%%%%%%%%%%%%%%%%%%%%%%%%%%%%%%%%%%%%%%%%%%%%%%%%%%%%%%%%%%%%%%%%%%%%%%%%%%%%%%%%%%%%%%
\newcommand{\bneg}{\blue{\neg}}
\newcommand{\beq}{\mathbin{\blue{=}}}
\newcommand{\bneq}{\mathbin{\blue{\neq}}}
\newcommand{\blor}{\mathbin{\blue{\lor}}}
\newcommand{\bland}{\mathbin{\blue{\land}}}
\newcommand{\bpi}{\blue{\Pi}}
\newcommand{\btop}{\blue{\top}}
\newcommand{\bbot}{\blue{\bot}}
\newcommand{\bleq}{{\blue{\Leq}}}

\section{Hintikka sets as defined by Steen~\cite{DBLP:phd/dnb/Steen18} \label{sec:steen}}
In the formulation of HOL as employed by Steen~\cite{DBLP:phd/dnb/Steen18},
the equality predicate, denoted $=^\tau$, for each type $\tau$, is assumed to be the only
logical connective present in the signature $\Sigma$, i.e.,
$\{ =^\tau \mid \tau \in \types \} \subseteq \Sigma$. All (potentially) remaining constant
symbols from $\Sigma$ are called parameters. Such signatures are also referred to as 
$\SigmaEq$.
A formulation of HOL based on equality as sole logical connective originates from Andrew's 
system $\mathcal{Q}_0$, cf.~\cite{andrews2002introduction} and the references therein.
The usual logical connectives are defined as follows
(technically our formulation is a modification of the one used by
Andrews~\cite{andrews2002introduction}, since the order of
terms in defining equations is swapped in many cases):\\

\begin{tabular}{lcl}
$\syntrue_o$ & := & ${=^o_{ooo}} \; {=^{ooo}_{o(ooo)(ooo)}} \; {=^o_{ooo}}$ \\
$\synfalse_o$ & := & $\left(\lambda P_o.\,P\right) =^{oo} \left(\lambda P_o.\, \syntrue\right)$ \\
$\neg_{oo}$ & := & $\lambda P_o.\, P =^o \synfalse$ \\
$\land_{ooo}$ & := & $\lambda P_o.\,\lambda Q_o.\, (\lambda F_{ooo}.\,F\;\syntrue\;\syntrue) =^{o(ooo)} (\lambda F_{ooo}.\,F\;P\;Q)$\\
$\lor_{ooo}$ & := & $\lambda P_o.\, \lambda Q_o.\, \neg \left(\neg P \land \neg Q \right)$ \\
$\Rightarrow_{ooo}$ & := & $\lambda P_o.\, \lambda Q_o.\, \neg P \lor Q$ \\
$\Leftrightarrow_{ooo}$ & := & $\lambda P_o.\, \lambda Q_o.\, P =^o Q$ \\
$\Pi^\tau_{o(o\tau)}$ & := & $\lambda P_{o\tau}.\,P =^{o\tau} \lambda X_\tau.\,\syntrue$
\end{tabular} \\

The (primitive and defined) connectives of this formulation of HOL are written in \blue{blue} in the following,
as are the properties below.

%%%%%%%%%%%%%%%%%%%%%%%%%%%%%%%%%%%%%%%%%%%%%%%%%%%%%%%%%%%%%%%%%%%%%%%%%%%%%%%%%%%%%%%%
\paragraph{Properties for acceptable Hintikka sets~\cite[Def. 3.15]{DBLP:phd/dnb/Steen18}}
\mbox{}\\

  $\blue{\vec{\nabla}_c}$:  $s \notin \mathcal{H}$ or $\blue{\neg} s \notin \mathcal{H}$.  
  
  $\blue{\vec{\nabla}_{\beta\eta}}$:  If $s \equiv_{\beta\eta} t$ and $s \in \mathcal{H}$, then $t \in \mathcal{H}$. 
  
  $\blue{\vec{\nabla}_{=}^{r}}$:  $(s \blue{\neq} s) \notin \mathcal{H}$. 
  
  $\blue{\vec{\nabla}_{=}^{s}}$:  If $u[s]_p \in \mathcal{H}$ and $s \blue{=} t \in \mathcal{H}$ then $u[t]_p \in \mathcal{H}$. 
  
  $\blue{\vec{\nabla}_{\mathfrak{b}}^+}$:  If $s \blue{=} t \in \mathcal{H}$, then $\{s,t\} \subseteq \mathcal{H}$ or
  $\{\blue{\neg} s, \blue{\neg} t\} \subseteq \mathcal{H}$. 
  
  $\blue{\vec{\nabla}_{\mathfrak{b}}^-}$:  If $s \blue{\neq} t \in \mathcal{H}$, then $\{s,\blue{\neg} t\} \subseteq \mathcal{H}$ or
  $\{\blue{\neg} s, t\} \subseteq \mathcal{H}$. 
  
  $\blue{\vec{\nabla}_{\mathfrak{f}}^+}$:  If $f_{\nu\tau} \blue{=} g_{\nu\tau} \in \mathcal{H}$, then $f \; s \blue{=} g \; s \in \mathcal{H}$ 
  for any closed term $s \in \Lambda^{\mathrm{c}}_\tau(\Sigma)$. 
  
  $\blue{\vec{\nabla}_{\mathfrak{f}}^-}$:  If $ f_{\nu\tau} \blue{\neq} g_{\nu\tau}  \in \mathcal{H}$, then $f \; w \blue{\neq} g  \; w \in \mathcal{H}$ 
  for some parameter $w \in \Sigma_\tau$. 
  
  $\blue{\vec{\nabla}_m}$:  If $s,t$ are atomic and $s, \blue{\neg} t \in \mathcal{H}$, then $s \blue{\neq} t  \in \mathcal{H}$. 
  
  $\blue{\vec{\nabla}_d}$:  If $h \; \overline{s^n} \blue{\neq} h \; \overline{t^n} \in \mathcal{H}$, then there is an
  $i$ with $1 \leq i \leq n$ such that $s^i \blue{\neq} t^i \in \mathcal{H}$.\\
  
\noindent The collection of all sets satisfying these properties is called $\mathfrak{H}$. Every element
$\H \in \mathfrak{H}$ is called acceptable.

\begin{definition}
A set $\H$ of formulas is called \emph{saturated} iff $s \in \H$ or $\blue{\neg}s \in \H$ for every closed formula $s$.
\end{definition}

\paragraph{Derived properties}

\begin{lemma}[Basic properties]\label{lemma:basic}
Let $\H \in \mathfrak{H}$. Then it holds that
\begin{enumerate}[(a)]
  \item \label{lemma:basic1} $\bot \notin \H$
  \item \label{lemma:basic2} $\blue{\neg}\top \notin \H$
  \item \label{lemma:basic3} $\top \blue{=} \bot \notin \H$
  \item \label{lemma:basic4}  If $s_o \blue{=} \top \in \H$ or $\top \blue{=} s_o \in \H$
                              ($s_o \blue{\neq} \bot \in \H$ or $\bot \blue{\neq} s_o \in \H$), then
                              $\{s_o,\top\} \subseteq \H$ ($\{s_o,\blue{\neg}\bot\} \subseteq \H$)
  \item \label{lemma:basic5}  If $s_o \blue{=} \bot \in \H$ or $\bot \blue{=} s_o \in \H$
                              ($s_o \blue{\neq} \top \in \H$ or $\top \blue{\neq} s_o \in \H$), then
                              $\{\blue{\neg}s_o,\blue{\neg}\bot\} \subseteq \H$ ($\{\blue{\neg}s_o,\top\} \subseteq \H$)
  \item \label{lemma:basic6} If $\blue{\neg}\bot \in \H$, then $\top \in \H$
  \item \label{lemma:basic7} If $\top \in \H$, then $\blue{\neg}\bot \in \H$
  \item \label{lemma:basic8} If $s \blue{=} t \in \H$ and $t \blue{=} u \in \H$, then $s \blue{=} u \in \H$.
\end{enumerate}
\end{lemma}
\begin{proof}
Let $\H \in \mathfrak{H}$ be an acceptable Hintikka set.
\begin{enumerate}[(a)]
  \item Assume $\bot \in \H$. By definition of $\bot$  it holds 
        $(\lambda P.\, P) \blue{=} (\lambda P.\, \top) \in \H$. Hence, by 
        $\blue{\vec{\nabla}_{\mathfrak{f}}^+}$ and
        $\blue{\vec{\nabla}_{\beta\eta}}$, 
        it follows that $w \blue{=} \top \in \H$ for any closed term $w$.
        Taking $w \equiv \blue{\neg}\top$ we obtain
        $\blue{\neg}\top \blue{=} \top \in \H$. But then, 
        $\blue{\vec{\nabla}_{\mathfrak{b}}^+}$ gives us a contradiction to 
        $\blue{\vec{\nabla}_{c}}$. Hence $\bot \notin \H$.
  \item Assume $\blue{\neg}\top \in \H$. By definition of
    $\bot$ it holds $(\blue{=} \blue{\neq} \blue{=})\in\H$, which contradicts 
        $\blue{\vec{\nabla}_{=}^r}$. % directly, since $\blue{\neg}\top \equiv (\blue{=} \blue{\neq} \blue{=})$.
        Hence, $\blue{\neg}\top \notin \H$.
  \item Assume $\top \blue{=} \bot \in \H$. Applying $\blue{\vec{\nabla}_{\mathfrak{b}}^+}$
        gives us that either $\{\top,\bot\} \subseteq \H$ or $\{\blue{\neg}\top, \blue{\neg}\bot\} \subseteq \H$.
        Either case is impossible by either \ref{lemma:basic1} or
        \ref{lemma:basic2} of this lemma. Hence, 
        $\top \blue{=} \bot \notin \H$.
  \item Let $s_o \blue{=} \top \in \H$ or $\top \blue{=} s_o \in \H$. In both cases it follows
        by $\blue{\vec{\nabla}_{\mathfrak{b}}^+}$ that either
        $\{s,\top\} \subseteq \H$ or $\{\blue{\neg}s, \blue{\neg}\top\} \subseteq \H$.
        Since the latter case contradicts \ref{lemma:basic2} from above, it follows that $\{s,\top\} \subseteq \H$.
        The negative cases are analogous using $\blue{\vec{\nabla}_{\mathfrak{b}}^-}$.
  \item Let $s_o \blue{=} \bot \in \H$ or $\bot \blue{=} s_o \in \H$. In both cases it follows
        by $\blue{\vec{\nabla}_{\mathfrak{b}}^+}$ that either
        $\{s,\bot\} \subseteq \H$ or $\{\blue{\neg}s, \blue{\neg}\bot\} \subseteq \H$.
        Since the former case contradicts \ref{lemma:basic1} from above, it follows that
        $\{\blue{\neg}s, \blue{\neg}\bot\} \subseteq \H$.
        The negative case is analogous using $\blue{\vec{\nabla}_{\mathfrak{b}}^-}$.  
  \item Let $\blue{\neg}\bot \in \H$. Then, by definition of $\bot$,
        $\big((\lambda P.\, P) \blue{\neq} (\lambda P.\, \top)\big) \in \H$.
        By $\blue{\vec{\nabla}_{\mathfrak{f}}^-}$ and $\blue{\vec{\nabla}_{\beta\eta}}$
        it holds that $p \blue{\neq} \top \in \H$ for some parameter $p$.
        By $\blue{\vec{\nabla}_{\mathfrak{b}}^-}$ it follows that either
        $\{p, \blue{\neg}\top\} \subseteq \H$
        or $\{\blue{\neg}p, \top\} \subseteq \H$.
        Since the former case is ruled out by \ref{lemma:basic2} from above, the latter case yields the desired result.
  \item Let $\top \in \H$, that is, $\blue{=}^o\blue{=}^{ooo}\blue{=}^o \in \H$.
        By $\blue{\vec{\nabla}_{\mathfrak{f}}^+}$ and $\blue{\vec{\nabla}_{\beta\eta}}$
        it follows that $(s_o \blue{=} t_o)\blue{=}(s_o\blue{=}t_o) \in \H$ for every
        two closed formulas $s,t$. 
        For $s \equiv t \equiv \blue{\neg}\bot$ it follows that 
        $(\blue{\neg}\bot \blue{=}
        \blue{\neg}\bot)\blue{=}(\blue{\neg}\bot\blue{=}\blue{\neg}\bot)
        \in \H$, 
        and hence, by $\blue{\vec{\nabla}_{\mathfrak{b}}^+}$, 
        either $\blue{\neg}\bot \blue{=} \blue{\neg}\bot \in \H$ or
        $\blue{\neg}\bot \blue{\neq} \blue{\neg}\bot \in \H$. Since the latter case is ruled out
        by $\blue{\vec{\nabla}_{=}^{r}}$, it follows that $\blue{\neg}\bot \blue{=} \blue{\neg}\bot \in \H$.
        Again, by $\blue{\vec{\nabla}_{\mathfrak{b}}^+}$, it follows that 
        either $\blue{\neg}\bot \in \H$ or
        $\blue{\neg}(\blue{\neg}\bot) \in \H$. The latter case is impossible by
        $\blue{\vec{\nabla}_{=}^r}$ since
        $\blue{\neg}(\blue{\neg} \bot) \equiv (\bot
        \blue{\neq} \bot)$ and hence $\blue{\neg}\bot \in \H$.
  \item Let $s \blue{=} t \in \H$ and $t \blue{=} u \in \H$. By $\blue{\vec{\nabla}_{=}^{s}}$ it
        follows directly that $s \blue{=} u \in \H$.
\end{enumerate}
\end{proof}

\begin{lemma}[Properties of usual connectives]\label{lemma:usualconnectives}
Let $\H \in \mathfrak{H}$. Then it holds that
\begin{enumerate}[(a)]
  \item \label{lemma:usualconnectives1} If $\blue{\neg}\blue{\neg} s_o \in \H$, then $s \in \H$
  \item \label{lemma:usualconnectives2} If $(s_o \blue{\lor} t_o) \in \H$, then $s \in \H$ or $t \in \H$.
  \item \label{lemma:usualconnectives3} If $(s_o \blue{\land} t_o) \in \H$, then $s \in \H$ and $t \in \H$.
  \item \label{lemma:usualconnectives4} If $\blue{\Pi^\alpha}\;F \in \H$, then $F\;s \in \H$ for every closed term $s$.
  \item \label{lemma:usualconnectives5} If $\blue{\neg\Pi^\alpha}\;F \in \H$, then $\blue{\neg}(F\;w) \in \H$ for some parameter $w \in \Sigma$.
\end{enumerate}
\end{lemma}
\begin{proof}
Let $\H \in \mathfrak{H}$ be an acceptable Hintikka set.
\begin{enumerate}[(a)]
  \item Let $\blue{\neg}\blue{\neg} s_o \in \H$. 
        By definition of $\blue{\neg}$ and
        $\blue{\vec{\nabla}_{\beta\eta}}$ it holds $(s \blue{\neq} \bot) \in \H$.
        Hence, by $\blue{\vec{\nabla}_{\mathfrak{b}}^-}$, either
          $\{s, \blue{\neg}\bot \} \subseteq \H$, or
          $\{\blue{\neg} s, \bot \} \subseteq \H$.
        As the latter case is impossible by Lemma~\ref{lemma:basic}\ref{lemma:basic1}, it follows that
        $\{s, \blue{\neg}\bot \} \subseteq \H$ and, in particular, that $s \in \H$.
  \item Let $(s_o \blue{\lor} t_o) \in \H$. By definition of
    $\blue{\lor}$, $\blue{\neg}$ and
        $\blue{\vec{\nabla}_{\beta\eta}}$  it holds
        $\big(\big(\lambda P.\, P\,\top\,\top\big) \blue{\neq}
          \lambda P.\,P\;(\blue{\neg}s)\;(\blue{\neg}t) \big) \in \H$. Hence, by $\blue{\vec{\nabla}_{\mathfrak{f}}^-}$ and
        $\blue{\vec{\nabla}_{\beta\eta}}$, 
        it follows that $(p\,\top\,\top) \blue{\neq} \big( p\;(\blue{\neg}s)\;(\blue{\neg}t)\big)$ for some
        parameter $p \in \Sigma$. By $\blue{\vec{\nabla}_{d}}$ either
          (i) $\top \blue{\neq} \blue{\neg}s \in \H$ or
          (ii) $\top \blue{\neq} \blue{\neg}t \in \H$. Hence, by 
        $\blue{\vec{\nabla}_{\mathfrak{b}}^-}$, applied to both cases,
        it holds that either
          (i) $\blue{\neg}\blue{\neg}s \in \H$, or
          (ii) $\blue{\neg}\blue{\neg}t \in \H$
        (because $\blue{\neg}\top \notin \H$ by Lemma~\ref{lemma:basic}\ref{lemma:basic2}).
        It follows that $s \in \H$ or $t \in \H$ by \ref{lemma:usualconnectives1} of this lemma.
  \item Let $(s_o \blue{\land} t_o) \in \H$. By definition of
    $\blue{\land}$ and
        $\blue{\vec{\nabla}_{\beta\eta}}$  it holds
        $(\lambda P.\, P\;\top\;\top) \blue{=} (\lambda P.\, P\;s\;t) \in\H$.
        Hence, by $\blue{\vec{\nabla}_{\mathfrak{f}}^+}$ and
        $\blue{\vec{\nabla}_{\beta\eta}}$, 
        it follows that $(w\;\top\;\top) \blue{=} (w\;s\;t) \in \H$ for every closed term $w$.
        By $\blue{\vec{\nabla}_{\beta\eta}}$, using
        $w \equiv \lambda x.\, \lambda y.\, x$ and 
        $w \equiv \lambda
        x.\, \lambda y.\, y$, it holds
         $\top \blue{=} s \in \H$ and $\top \blue{=} t \in \H$, respectively.
        Application of Lemma~\ref{lemma:basic}\ref{lemma:basic4} yields the desired result.
  \item Let $\blue{\Pi^\alpha}\;s \in \H$. By definition of
    $\blue{\Pi^\alpha}$ and
        $\blue{\vec{\nabla}_{\beta\eta}}$  it holds
        $\big(s \blue{=} \lambda x.\, \top \big) \in\H$. Hence, by 
        $\blue{\vec{\nabla}_{\mathfrak{f}}^+}$ and $\blue{\vec{\nabla}_{\beta\eta}}$,
        it follows that $s\;t \blue{=} \top \in \H$ for every closed term $t$.
        Application of Lemma~\ref{lemma:basic}\ref{lemma:basic4} yields the desired result.
  \item Let $\blue{\neg\Pi^\alpha}\;s \in \H$.  By definition of
    $\blue{\Pi^\alpha}$ and
        $\blue{\vec{\nabla}_{\beta\eta}}$  it holds that 
        $\big(s \blue{\neq} (\lambda x.\, \top)\big)\in\H$. Hence, by
        $\blue{\vec{\nabla}_{\mathfrak{f}}^-}$ and $\blue{\vec{\nabla}_{\beta\eta}}$,
        it follows that $(s\;p) \blue{\neq} \top \in \H$ for some parameter $p$.
        Application of Lemma~\ref{lemma:basic}\ref{lemma:basic5} yields the desired result.
\end{enumerate}
\end{proof}

\begin{lemma}[Properties of Leibniz equality]\label{lemma:additional}
Let $\H \in \mathfrak{H}$. Then it holds that
\begin{enumerate}[(a)]
  \item \label{lemma:additional2} If $s \blue{\Leq} t \in \H$, then $s \blue{=} t \in \H$.
  \item \label{lemma:additional3} If $\blue{\neg}(s \blue{\Leq} t) \in \H$, then $s \blue{\neq} t \in \H$.
  \item \label{lemma:additional4} $\blue{\neg}(s \blue{\Leq} s) \notin \H$.
  \item \label{lemma:additional5} If $u[s]_p \in \H$ and $s \blue{\Leq} t \in \H$, then $u[t]_p \in \H$.
  \item \label{lemma:additional6} If $s \blue{\Leq} t \in \H$, then $t \blue{\Leq} s \in \H$.
  \item \label{lemma:additional7} If $s \blue{\Leq} t \in \H$ and $t \blue{\Leq} u \in \H$, then $s \blue{\Leq} u \in \H$.
\end{enumerate}
\end{lemma}
\begin{proof}
Let $\H \in \mathfrak{H}$ be an acceptable Hintikka set.
\begin{enumerate}[(a)]
\item Let $(s \blue{\Leq} t) \in \H$. By definition of $\blue{\Leq}$
  and $\blue{\vec{\nabla}_{\beta\eta}}$ we have
  $\big(\lambda P.\, (P\;s) \blue{\Rightarrow} (P\;t)\big) \blue{=}
  \big(\lambda P.\,\top\big) \in \H$,
  and hence, by $\blue{\vec{\nabla}_{\mathfrak{f}}^+}$ and
  $\blue{\vec{\nabla}_{\beta\eta}}$, it holds that
  $\big((w\;s) \blue{\Rightarrow} (w\;t)\big) \blue{=} \top \in \H$
  for every closed term $w$.  Then,
  $(w\;s) \blue{\Rightarrow} (w\;t) \in \H$ by
  Lemma~\ref{lemma:basic}\ref{lemma:basic4}.  By definition of
  $\blue{\Rightarrow}$ and $\blue{\vec{\nabla}_{\beta\eta}}$ it holds
  that $\neg(w\;s) \blue{\lor} (w\;t) \in \H$ and hence, by
  Lemma~\ref{lemma:usualconnectives}\ref{lemma:usualconnectives2}, that
  $\neg(w\;s) \in \H$ or $(w\;t) \in \H$.  For
  $w \equiv (\lambda X.\, s \blue{=} X)$ it follows by
  $\blue{\vec{\nabla}_{\beta\eta}}$ that $\neg(s \blue{=} s) \in \H$
  or $(s \blue{=} t) \in \H$.  Since the former case contradicts
  $\blue{\vec{\nabla}_{=}^r}$, it follows that
  $(s \blue{=} t) \in \H$.
  
  \item Let $\blue{\neg}(s \blue{\Leq} t) \in \H$. By definition of $\blue{\Leq}$
  and $\blue{\vec{\nabla}_{\beta\eta}}$ we have
   $\big(\lambda P.\, (P\;s) \blue{\Rightarrow} (P\;t)\big) \blue{\neq} \big(\lambda P.\,\top\big) \in \H$.
  Hence, by $\blue{\vec{\nabla}_{\mathfrak{f}}^-}$ and
  $\blue{\vec{\nabla}_{\beta\eta}}$, 
  it holds that $\big((p\;s) \blue{\Rightarrow} (p\;t)\big) \blue{\neq} \top \in \H$ for some
  parameter $p$.
  By Lemma~\ref{lemma:basic}\ref{lemma:basic5} it follows that
  $\blue{\neg}\big((p\;s) \blue{\Rightarrow} (p\;t)\big) \in
  \H$. 
  Then, by Lemma~\ref{lemma:usualconnectives}\ref{lemma:usualconnectives1} and \ref{lemma:usualconnectives}\ref{lemma:usualconnectives3}, it follows that
  $\blue{\neg}\blue{\neg}(p\;s) \in \H$ and $\blue{\neg}(p\;t) \in \H$. Moreover, $(p\;s) \in \H$ by
  Lemma~\ref{lemma:usualconnectives}\ref{lemma:usualconnectives1}.
  By $\blue{\vec{\nabla}_{{m}}}$ it then follows that $(p\;s)
  \blue{\neq} (p\;t) \in \H$, and finally, by
  $\blue{\vec{\nabla}_{{d}}}$, 
  that $s \blue{\neq} t \in \H$.
  
  \item Assume $\blue{\neg}(s \blue{\Leq} s) \in \H$.
        By~\ref{lemma:additional3} above it follows that $s \blue{\neq} s \in \H$
        which contradicts $\blue{\vec{\nabla}_{=}^r}$.
        Hence, $\blue{\neg}(s \blue{\Leq} s) \notin \H$.
  
  \item Let $u[s]_p \in \H$ and $s \blue{\Leq} t \in \H$. By~\ref{lemma:additional2} above it holds that 
        $s \blue{=} t \in \H$ and thus by $\blue{\vec{\nabla}_=^s}$ it follows
        that $u[t]_p \in \H$.
  
  \item Let $(s \blue{\Leq} t) \in \H$. By definition of $\blue{\Leq}$
  and $\blue{\vec{\nabla}_{\beta\eta}}$ 
  we have $\big(\lambda P.\, (P\;s) \blue{\Rightarrow} (P\;t)\big) \blue{=} \big(\lambda P.\,\top\big) \in \H$.
  Hence, by $\blue{\vec{\nabla}_{\mathfrak{f}}^+}$ and $\blue{\vec{\nabla}_{\beta\eta}}$,
  it holds that $\big((w\;s) \blue{\Rightarrow} (w\;t)\big) \blue{=} \top \in \H$ for every
  closed term $w$. Then, $(w\;s) \blue{\Rightarrow} (w\;t) \in \H$ by Lemma~\ref{lemma:basic}\ref{lemma:basic4}.
  By definition of $\blue{\Rightarrow}$ and $\blue{\vec{\nabla}_{\beta\eta}}$ it holds that
  $\neg(w\;s) \blue{\lor} (w\;t) \in \H$ and hence by Lemma~\ref{lemma:usualconnectives}\ref{lemma:usualconnectives2} 
  that $\neg(w\;s) \in \H$ or $(w\;t) \in \H$.
  For $w \equiv (\lambda X.\, t \blue{\Leq} s)$, it follows by $\blue{\vec{\nabla}_{\beta\eta}}$
  that $\neg(t \blue{\Leq} s) \in \H$ or $(t \blue{\Leq} s) \in \H$. Assume
  $\neg(t \blue{\Leq} s) \in \H$. Since by Lemma~\ref{lemma:additional}\ref{lemma:additional2}
  it holds that $s \blue{=} t \in \H$, it follows by $\blue{\vec{\nabla}_=^s}$
  that $\neg(t \blue{\Leq} t) \in \H$, which
  contradicts~\ref{lemma:additional4} above. Hence, $(t \blue{\Leq} s) \in \H$.
  
  \item Let $(s \blue{\Leq} t) \in \H$ and $(t \blue{\Leq} u) \in \H$.
        By~\ref{lemma:additional5} above it follows that 
        $(s \blue{\Leq} u) \in \H$.
\end{enumerate}
\end{proof}

\begin{lemma}[Sufficient conditions for saturatedness]\label{lemma:sufficient}
Let $\H \in \mathfrak{H}$. Then it holds that
\begin{enumerate}[(a)]
  \item \label{lemma:sufficient1} If $\top \in \H$, then $\H$ is saturated.
  \item \label{lemma:sufficient2} If $\blue{\neg} s \in \H$ for some closed term $s$, then $\H$ is saturated.
  \item \label{lemma:sufficient3} If $s \blue{\lor} t \in \H$ for some closed terms $s,t$, then $\H$ is saturated.
  \item \label{lemma:sufficient4} If $s \blue{\land} t \in \H$ for some closed terms $s,t$, then $\H$ is saturated.
  \item \label{lemma:sufficient5} If $\blue{\Pi}^\tau\;P \in \H$ for some closed term $P_{\tau\to o}$, then $\H$ is saturated.
  \item \label{lemma:sufficient6} If $s \blue{\Leq} t \in \H$ for some closed terms $s,t$, then $\H$ is saturated.
\end{enumerate}
\end{lemma}
\begin{proof}
Let $\H \in \mathfrak{H}$ be an acceptable Hintikka set.
\begin{enumerate}[(a)] 
  \item Let $\top \in \H$, that is, $\blue{=}^o\blue{=}^{ooo}\blue{=}^o \in \H$.
        By $\blue{\vec{\nabla}_{\mathfrak{f}}^+}$ and $\blue{\vec{\nabla}_{\beta\eta}}$
        it follows that $(s_o \blue{=} t_o)\blue{=}(s_o\blue{=}t_o) \in \H$ for every
        two closed formulas $s,t$. For $s \equiv t \equiv c$ for some closed term $c$,
        it follows that 
        $(c \blue{=} c)\blue{=}(c\blue{=}c) \in \H$
        and thus, by $\blue{\vec{\nabla}_{\mathfrak{b}}^+}$ and $\blue{\vec{\nabla}_{=}^{r}}$,
        it holds that $c \blue{=} c \in H$. By $\blue{\vec{\nabla}_{\mathfrak{b}}^+}$ 
        it follows that $c \in \H$ or $\blue{\neg}c \in \H$.
        Hence, $\H$ is saturated.
       
  \item If $\neg s \in \H$ for some closed term $s$, then $s \blue{=} \bot \in \H$.
        By $\blue{\vec{\nabla}_{\mathfrak{b}}^+}$, it follows that
        either $\{s,\bot\} \subseteq \H$ or $\{\blue{\neg}s,\blue{\neg}\bot\} \subseteq \H$.
        Since the former case is ruled out by Lemma~\ref{lemma:basic}\ref{lemma:basic1}, 
        it follows that $\blue{\neg}\bot \in \H$.
        By Lemma~\ref{lemma:basic}\ref{lemma:basic6}  it follows that $\top \in \H$ and by \ref{lemma:sufficient1} above
        it follows that $\H$ is saturated.
        
  \item If $s \blue{\lor} t \in \H$ for some closed terms $s,t$, then
    by definition of $\blue{\lor}$ and
    $\blue{\vec{\nabla}_{\beta\eta}}$ we have 
        $\blue{\neg}(\blue{\neg}s \blue{\land} \blue{\neg}t) \in \H$. An application of
        \ref{lemma:sufficient2} yields the desired result.
        
  \item If $s \blue{\land} t \in \H$ for some closed terms $s,t$, then
    by definition of $\blue{\land}$ and
    $\blue{\vec{\nabla}_{\beta\eta}}$ it holds 
        $(\lambda g.\, g\;s\;t) \blue{=} (\lambda g.\, g\;\top\;\top) \in \H$.
        By $\blue{\vec{\nabla}_{\mathfrak{f}}^+}$ and $\blue{\vec{\nabla}_{\beta\eta}}$
        it follows that $s \blue{=} \top \in \H$ (take $\lambda x.\, \lambda y.\, x$).
        By $\blue{\vec{\nabla}_{\mathfrak{b}}^+}$, it follows that
        either $\{s,\top\} \subseteq \H$ or $\{\blue{\neg}s,\blue{\neg}\top\} \subseteq \H$.
        Since the latter case is ruled out by Lemma~\ref{lemma:basic}\ref{lemma:basic2}, 
        it follows that $\top \in \H$.
        An application of \ref{lemma:sufficient1} above yields the desired result.
        
  \item If $\blue{\Pi}^\tau\;s \in \H$ for some closed terms $s$, then
    by definition
    of $\blue{\Pi}^\tau$ and
    $\blue{\vec{\nabla}_{\beta\eta}}$ it holds that
        $s \blue{=} (\lambda x.\, \top) \in \H$.
        By $\blue{\vec{\nabla}_{\mathfrak{f}}^+}$ and $\blue{\vec{\nabla}_{\beta\eta}}$
        it follows that $(s\;w) \blue{=} \top \in \H$ for every closed term $w$.
        By $\blue{\vec{\nabla}_{\mathfrak{b}}^+}$, it follows that
        either $\{(s\;w),\top\} \subseteq \H$ or $\{\blue{\neg}(s\;w),\blue{\neg}\top\} \subseteq \H$.
        Since the latter case is ruled out by Lemma~\ref{lemma:basic}\ref{lemma:basic2}, 
        it follows that $\top \in \H$.
        An application of \ref{lemma:sufficient1} above yields the desired result.

  \item Let $s \blue{\Leq} t \in \H$. By definition
    of $\blue{\Pi}^\tau$ and
    $\blue{\vec{\nabla}_{\beta\eta}}$ it holds that 
        $\blue{\Pi}\big(\lambda P.\, (P\;s) \blue{\Rightarrow} (P\;t)\big) \in \H$.
        An application of \ref{lemma:sufficient5} yields the desired result.

\end{enumerate}
\end{proof}

\begin{corollary}\label{corollary:equalitysaturated}
Let $\H \in \mathfrak{H}$ and let $s \blue{\neq} t \in \H$ or $\blue{\neg}(s \blue{\Leq} t) \in \H$
for some closed terms $s,t$.
Then, $\H$ is saturated.
\end{corollary}
\begin{proof}
As $(s \blue{\neq} t) \equiv \blue{\neg}(s \blue{=} t)$, both cases are a special instance of 
Lemma~\ref{lemma:sufficient}\ref{lemma:sufficient2}.
\end{proof}

\begin{lemma}[Saturated sets properties]\label{lemma:saturated}
Let $\H \in \mathfrak{H}$ and let $\H$ be saturated. Then it holds that
\begin{enumerate}[(a)]  
  \item \label{lemma:saturated1} If $s \blue{=} t \in \H$, then $s \blue{\Leq} t \in \H$.
%  \item \label{lemma:saturated2} If $s \blue{\neq} t \in \H$, then $\blue{\neg}(s \blue{\Leq} t) \in \H$.
  \item \label{lemma:saturated3} If $s \blue{=} t \in \H$ then $t \blue{=} s \in \H$.
  \item \label{lemma:saturated4} $s \blue{=} s \in \H$ for every closed term $s$.
  \item \label{lemma:saturated5} $s \blue{\Leq} s \in \H$ for every closed term $s$.
\end{enumerate}
\end{lemma}
\begin{proof}
Let $\H \in \mathfrak{H}$ and let $\H$ be saturated.
\begin{enumerate}[(a)]  
  \item Let $s \blue{=} t \in \H$ and assume $s \blue{\Leq} t \notin \H$.
        Since $\H$ is saturated we  have $\blue{\neg}(s \blue{\Leq} t) \in \H$.
        Then, by Lemma~\ref{lemma:additional}\ref{lemma:additional3}, 
        it follows that $s \blue{\neq} t \in \H$, and thus $\{s \blue{=} t, s \blue{\neq} t\} \subseteq \H$,
        which contradicts $\blue{\vec{\nabla}_c}$.
        Hence, $s \blue{\Leq} t \in \H$.
        
%  \item Let $s \blue{\neq} t \in \H$ and assume $\blue{\neg}(s \blue{\Leq} t) \notin \H$.
%        Since $\H$ is saturated we have that $(s \blue{\Leq} t) \in \H$.
%        Then, by Lemma~\ref{lemma:additional}\ref{lemma:additional2}, it follows that $s
%        \blue{=} t \in \H$, and thus $\{s \blue{=} t, s \blue{\neq}
%        t\} \subseteq \H$, which contradicts $\blue{\vec{\nabla}_c}$.
%        Hence, $\blue{\neg}(s \blue{\Leq} t) \in \H$.
%  
  \item Let $s \blue{=} t \in \H$ and assume $t \blue{=} s \notin \H$.
        Since $\H$ is saturated we have $t \blue{\neq} s \in \H$. 
        Then, by $\blue{\vec{\nabla}_=^s}$, it follows that $t \blue{\neq} t \in \H$
        which contradicts $\blue{\vec{\nabla}_=^r}$.
        Hence, $t \blue{=} s \in \H$.
        
  \item Let $s$ be a closed term of some type and assume that $s \blue{=} s \notin \H$.
        Since $\H$ is saturated we have that $s \blue{\neq} s \in \H$. Since this
        contradicts $\blue{\vec{\nabla}_=^r}$ it follows that $s \blue{=} s \in \H$.
        
  \item Let $s$ be a closed term of some type and assume that $s \blue{\Leq} s \notin \H$.
        Since $\H$ is saturated we have that $\blue{\neg}(s \blue{\neq} s) \in \H$. Since this
        is impossible by Lemma~\ref{lemma:additional}\ref{lemma:additional4} it follows that $s \blue{\Leq} s \in \H$.
\end{enumerate}
\end{proof}

\begin{lemma}[Properties of negated equalities]\label{lemma:negatedequality}
Let $\H \in \mathfrak{H}$. It holds that
\begin{enumerate}[(a)]
  \item \label{lemma:negatedequality1} If $s \bneq t \in \H$, then
        $t \bneq s \in \H$.
  \item \label{lemma:negatedequality2} If $\bneg(s \bleq t) \in \H$, then
        $\bneg(t \bleq s) \in \H$
  \item \label{lemma:negatedequality3} If $s \bneq t \in \H$,
        then $\bneg(s \bleq t) \in \H$.
\end{enumerate}
\end{lemma}
\begin{proof}
Let $\H \in \mathfrak{H}$.
\begin{enumerate}[(a)]
  \item Let $s \bneq t \in \H$ and assume that $t \bneq s \notin \H$.
        By Corollary~\ref{corollary:equalitysaturated}
        it follows that $\H$ is saturated and hence we have that
        $t \beq s \in \H$.
        By saturation and Lemma~\ref{lemma:saturated}\ref{lemma:saturated3}
        it follows that $s \beq t \in \H$, and thus $\{s \bneq t, s \beq t \} \subseteq \H$, 
        which contradicts $\blue{\vec{\nabla}_c}$.
        Hence, $t \bneq s \in \H$.
  \item Let $\bneg(s \bleq t) \in \H$ and assume $\bneg(t \bleq s) \notin \H$.
        By Corollary~\ref{corollary:equalitysaturated}
        it follows that $\H$ is saturated and hence we have that
        $t \bleq s \in \H$.
        Then, by Lemma~\ref{lemma:additional}\ref{lemma:additional5}, it follows
        that $s \bleq t \in \H$, and thus $\{\bneg(s \bleq t), s \bleq t\} \subseteq \H$,
        which contradicts $\blue{\vec{\nabla}_c}$.
        Hence, $\bneg(t \bleq s) \in \H$.
  \item Let $s \bneq t \in \H$ and assume $\bneg(s \bleq t) \notin \H$.
        By Corollary~\ref{corollary:equalitysaturated}
        it follows that $\H$ is saturated and hence we have that  
        $(s \bleq t) \in \H$.
        Then, by Lemma~\ref{lemma:additional}\ref{lemma:additional2}, it follows
        that $s \beq t \in \H$, and thus $\{s \beq t, s \bneq t\} \subseteq \H$,
        which contradicts $\blue{\vec{\nabla}_c}$.
        Hence, $\bneg(s \bleq t) \in \H$.
\end{enumerate}
\end{proof}

\begin{definition}[Leibniz-free]
Let $S$ be a set of formulas. $S$ is called \emph{Leibniz-free} iff
$s \blue{\Leq} t \notin S$ for any terms $s,t$.
\end{definition}

\begin{corollary}[Impredicativity Gap]\label{corollary:gap}
Let $\H \in \mathfrak{H}$. $\H$ is saturated or Leibniz-free.
\end{corollary}
\begin{proof}
Assume that $\H$ is not Leibniz-free. Then, there exists some formula
$s \blue{\Leq} t \in \H$. An application of Lemma~\ref{lemma:sufficient}\ref{lemma:sufficient6}
yields the desired result.
\end{proof}

\paragraph{Summary of properties of equality and Leibniz-equality.}
The following table contains an overview of the implied properties of $\blue{=}$ and $\blue{\Leq}$,
respectively. A property that holds unconditionally is marked with $\checkmark$, 
a property that holds for saturated Hintikka sets is marked with $\mathsf{sat.}$.\\

\begin{center}
\begin{tabular}{l||c|c}
\textbf{Property} & $\star \equiv \blue{=}$ & $\star \equiv \blue{\Leq}$ \\
\hline
$s \star s \in \H$ & $\mathsf{sat.}$ & $\mathsf{sat.}$ \\
$\blue{\neg}(s \star s) \notin \H$ & $\checkmark$ & $\checkmark$ \\
If $s \star t \in \H$ and $t \star u \in \H$, then $s \star u \in \H$ & $\checkmark$ & $\checkmark$ \\
If $s \star t \in \H$, then $t \star s \in \H$ & $\mathsf{sat.}$ & $\checkmark$ \\
If $u[s]_p \in \H$ and $s \star t \in \H$, then $u[t]_p \in \H$ & $\checkmark$ & $\checkmark$
\end{tabular}
\end{center}

%%%%%%%%%%%%%%%%%%%%%%%%%%%%%%%%%%%%%%%%%%%%%%%%%%%%%%%%%%%%%%%%%%%%%%%%%%%%%%%%%%%%%%%%
%%%%%%%%%%%%%%%%%%%%%%%%%%%%%%%%%%%%%%%%%%%%%%%%%%%%%%%%%%%%%%%%%%%%%%%%%%%%%%%%%%%%%%%%
%%%%%%%%%%%%%%%%%%%%%%%%%%%%%%%%%%%%%%%%%%%%%%%%%%%%%%%%%%%%%%%%%%%%%%%%%%%%%%%%%%%%%%%%
\section{Reduction of $\mathfrak{H}$ (Steen) to
  $\mathfrak{Hint}_{\beta\mathfrak{f}\mathfrak{b}}$ (Brown)} \label{sec:reduction}
\newcommand{\HH}{\H_{\red{\triangleright}}}
\newcommand{\HHH}{\HH^\sharp}

We reduce the notion of Hintikka sets of Steen to the notion of Hintikka
sets by Brown. 
Conceptually, every formula from a set $\H \in \mathfrak{H}$  is first translated 
to its red equivalent under a signature $\red{\Sigma} \subseteq \{ \rneg \} \cup \{ \req^\tau \mid \tau \in \types \}$
containing no further primitive logical connectives.
Note that this involves mapping a primitive connective to a primitive connective (i.e.\ $\beq$ to $\req$)
as well as mapping a defined connective to a primitive connective (i.e.\ $\bneg$ to $\rneg$).
In a second step, the red formulas are translated into their meta-counterparts, i.e.\ into
external propositions.

Formally, we define the mapping as follows: Let $\red{\Sigma}$ be
given as
introduced above. We define $\HH :=
\{(s[\bneg\backslash\rneg])[\beq\backslash\req]| s\in
\mathfrak{H}\}$, where $s[\blue{l}\backslash \red{r}]$ denotes that
term in $\Lambda^c(\red{\Sigma})$ that is obtained by replacing all
occurrences of  $\blue{l}$ in $s$ by $\red{r}$. Obviously, if
$s\in\Lambda^c_o(\blue{\Sigma})$ then
$(s[\bneg\backslash\rneg])[\beq\backslash\req]\in\Lambda^c_o(\red{\Sigma})$,
and thus $\HH\subseteq
\Lambda^c_o(\red{\Sigma})$.

% To that end, a translation scheme is employed as follows:
% Let $\HH \subseteq \Lambda^c(\red{\Sigma})$ be the set that is constructed from a set $\H \in \mathfrak{H}$ of closed formulas
% by replacing all occurrences of \blue{blue} connectives (SAY EXACTLY WHICH CONNECTIVES THESE ARE?) by their corresponding \red{red} connective. 
% The (recursively) translated formula associated with $s \in \H$ is written $\red{s}$ (in red).

As the Hintikka properties of Brown are defined in terms of external propositions, we proceed
by constructing a set $\HHH \subseteq prop(\red{\Sigma})$ from $\HH$ by enriching $\HH$ with
its external counterparts:
To that end, for any term $s_o \in \Lambda_o(\red{\Sigma})$, we denote by $s^\sharp \in prop(\red{\Sigma})$
the term that is constructed by replacing the primitive connective at head position by its external equivalent\footnote{
This corresponds to the original definition of $s^\sharp$ by Brown~\cite[Def. 2.1.38]{brown07}.
}, i.e.\
we have that $(\rneg s)^\sharp$  is equal to  $\rdneg s$ and that $(s
\req^o t)^\sharp$ is equal
 to $s \rdeq^o t$. If there is no connective at head position, the term is left unchanged, i.e.,
$(p\; \overline{s^n})^\sharp$ equals $p \; \overline{s^n}$ and $(X\; \overline{s^n})^\sharp$
equals $X \; \overline{s^n}$ whenever $p \in \red{\Sigma}$ is a parameter and $X$ is a variable.

Then, $\HHH$ is defined inductively as follows. For $\H \in \mathfrak{H}$, let $\HHH$ be the smallest
set of external propositions over $\red{\Sigma}$ such that
\begin{enumerate}[(1)]
  \item if $s_o \in \H$, then $\red{s} \in \HHH$, and
  \item if $s_o \in \H$, then $\red{s}^\sharp \in \HHH$, and
  \item if $\rdneg \red{s} \in \HHH$, then $\rdneg \red{s}^\sharp \in \HHH$.
\end{enumerate}

Intuitively, this translation adopts the translated object-level terms of $\H$ (clause (1)),
and additionally augments the set corresponding meta-level terms by replacing
connectives at head positions by its meta-connective equivalent (clause (2)).
Also, connectives directly under a (meta-)negation are considered (clause (3)).\footnote{
Clauses (2) and (3) intuitively reflect property $\red{\vec{\nabla}^\sharp}$
of Hintikka sets by Brown~\cite[Def. 5.5.1]{brown07}.
}
Deep translations are then implicitly provided by the Hintikka closure properties.

We provide simple examples of unsaturated Hintikka sets $\H$, $\HH$
and $\HHH$:
\begin{itemize}
\item  Let $\H := \{s | s \equiv_{\beta\eta} t \text{ for } t \in \{a
\beq^\iota b, a\beq^\iota  a, b \beq^\iota b, b
\beq^\iota a\}\}$. Then $\HH = \{s | s \equiv_{\beta\eta} t \text{ for } t \in \{a
\req^\iota b, a\req^\iota  a, b \req^\iota b, b
\req^\iota a\}\}$ and $\HHH = \HH \cup \{s | s \equiv_{\beta\eta} t \text{ for } t \in \{a
\rdeq^\iota b, a\rdeq^\iota  a, b \rdeq^\iota b, b
\rdeq^\iota a\}\}.$
\item  Let $\H := \{s | s \equiv_{\beta\eta} t \text{ for } t \in \{a
\beq^\iota p(\bneg), a\beq^\iota  a, p(\bneg) \beq^\iota p(\bneg), p(\bneg)
\beq^\iota a\}\}$. Then $\HH = \{s | s \equiv_{\beta\eta} t \text{ for } t \in \{a
\req^\iota p(\rneg), a\req^\iota  a, p(\rneg) \req^\iota p(\rneg), p(\rneg)
\req^\iota a\}\}$ and $\HHH = \HH \cup \{s | s \equiv_{\beta\eta} t \text{ for } t \in \{a
\rdeq^\iota p(\rneg), a\rdeq^\iota  a, p(\rneg) \rdeq^\iota p(\rneg), p(\rneg)
\rdeq^\iota a\}\}.$
\end{itemize}

%\fbox{.... insert example (TODO) .... We should take an example in the
%  spirit of the reviewers comment. Comment: saturation not an issue anymore.}

We now show that if $\H \in \mathfrak{H}$ then
$\HHH \in \mathfrak{Hint}_{\beta\mathfrak{f}\mathfrak{b}}$
(i.e., $\HHH$ fulfils all $\red{\vec{\nabla}}$ from \S\ref{sec:chad}).

\begin{theorem}[Reduction of {$\mathfrak{H}$ to $\mathfrak{Hint}_{\beta\mathfrak{f}\mathfrak{b}}$}]
If $\H \in \mathfrak{H}$, then there exists an extensional Hintikka set
$\H^\prime \in \mathfrak{Hint}_{\beta\mathfrak{f}\mathfrak{b}}$
such that $\HHH \equiv \mathcal{H}^\prime$.  \label{thm:reduction}
\end{theorem}
\begin{proof}
Let $\H \in \mathfrak{H}$ be a Hintikka set according
\S\ref{sec:steen}, i.e., fulfilling all $\blue{\vec{\nabla}}$
properties. 

\noindent We verify every $\red{\vec{\nabla}}$-property for $\HHH$:
\begin{description}
  \item $\red{\vec{\nabla}_c}$: Assume that both $\red{s} \in \HHH$ and $\rdneg \red{s} \in \HHH$.
        Then, by definition, $s \in H$ and $\bneg s \in H$. As this contradicts $\blue{\vec{\nabla}_c}$,
        it follows that $\red{s} \notin \HHH$ or $\rdneg \red{s} \notin \HHH$.

  \item $\red{\vec{\nabla}_{\beta\eta}}$: Let $\red{s} \in \HHH$. Then, by definition, $s \in \H$.
        Since $s \equiv_{\beta\eta} s^\downarrow$ it follows by $\blue{\vec{\nabla}_{\beta\eta}}$
        that $s^\downarrow \in \H$. By definition, we then have $\red{s}^{\downarrow} \in \HHH$.

  \item $\red{\vec{\nabla}_\bot}$: $\rdneg \rdtop \notin \HHH$ is vacuously true, as 
        $\rdtop$ is never the result of any term translation.
        
  \item $\red{\vec{\nabla}_\neg}$: Let $\rdneg\rdneg s \in \HHH$.
        Then, by definition, $\bneg\bneg s \in \H$. By Lemma~\ref{lemma:usualconnectives}\ref{lemma:usualconnectives1}
        it follows that $s \in \H$ and hence $\red{s} \in \HHH$.

  \item $\red{\vec{\nabla}_\lor}$: Vacuously true, as $\rdlor$ is never the result of any term translation.

  \item $\red{\vec{\nabla}_\land}$: Vacuously true, as $\rdlor$ is never the result of any term translation.

  \item $\red{\vec{\nabla}_\forall}$: Vacuously true, as $\rdforall$ is never the result of any term translation.

  \item $\red{\vec{\nabla}_\exists}$: Vacuously true, as $\rdforall$ is never the result of any term translation.

  \item $\red{\vec{\nabla}^\sharp}$: Holds by definition of $\HHH$.

  \item $\red{\vec{\nabla}_{\mathit{m}}}$: Let $\rdneg\big(\red{p}\;\overline{\red{s^n}}\big) \in \HHH$
        and let $\big(\red{p}\;\overline{\red{t^n}}\big) \in \HHH$ for some parameter $p$.
        Then, by definition, $\bneg\big(p\;\overline{s^n}\big) \in \H$
        and $\big({p}\;\overline{{t^n}}\big) \in \H$. By $\blue{\vec{\nabla}_{\mathit{m}}}$
        it follows that $\big(p\;\overline{t^n}\big) \bneq \big(p\;\overline{s^n}\big) \in \H$.
        By $\blue{\vec{\nabla}_{\mathit{d}}}$ it then follows that there is some $i$, $1 \leq i \leq n$, 
        such that $t^i \bneq s^i \in \H$. 
        By Lemma~\ref{lemma:negatedequality}\ref{lemma:negatedequality1} it follows
        that $s^i \bneq t^i \in \H$ and hence $\rdneg(\red{s^i} \rdeq \red{t^i}) \in \HHH$.

  \item $\red{\vec{\nabla}_{\mathit{dec}}}$: Let $p$ be a parameter and
        let $\rdneg\big((\red{p}\;\overline{\red{s^n}}) \rdeq^\iota (\red{p}\;\overline{\red{t^n}})\big) \in \HHH$.
        Then, by definition, $\bneg\big(({p}\;\overline{{s^n}}) \beq^\iota ({p}\;\overline{{t^n}})\big) \in \H$.
        By $\blue{\vec{\nabla}_{\mathit{d}}}$ it follows that there is some $i$, $1 \leq i \leq n$, 
        such that $\bneg ( s^i \beq t^i) \in \H$ and hence $\rdneg ( \red{s^i} \rdeq \red{t^i}) \in \HHH$.

  \item $\red{\vec{\nabla}_{\mathfrak{b}}}$: Let $\rdneg(\red{s} \rdeq^o \red{t}) \in \HHH$.
        Then, by definition, $\bneg({s} \beq^o {t}) \in \H$. By $\blue{\vec{\nabla}_{\mathfrak{b}}^-}$
        it follows that $\{s, \bneg t \} \subseteq \H$ or $\{\bneg s, t \} \subseteq \H$.
        Hence we have that $\{\red{s}, \rdneg \red{t} \} \subseteq \HHH$ or $\{\rdneg \red{s}, \red{t} \} \subseteq \HHH$.

  \item $\red{\vec{\nabla}_{\mathfrak{f}}}$: Let $\rdneg(\red{f} \rdeq^{\nu\tau} \red{g}) \in \HHH$.
        Then, by definition, $\bneg({f} \beq^{\nu\tau} {g}) \in \H$. By $\blue{\vec{\nabla}_{\mathfrak{f}}^-}$
        it follows that $\bneg({f}\;p\beq^{\nu} {g}\;p) \in \H$ for some parameter $p_\tau$.
        Hence we have that $\rdneg(\red{f}\;\red{p}\rdeq^{\nu} \red{g}\;\red{p}) \in \HHH$.

  \item $\red{\vec{\nabla}_{=}^o}$: Let $\red{s} \rdeq^o \red{t} \in \HHH$.
        Then, by definition, $s \beq^o t \in \H$. By $\blue{\vec{\nabla}_{\mathfrak{b}}^+}$
        it follows that $\{s, t \} \subseteq \H$ or $\{\bneg s, \bneg t \} \subseteq \H$.
        Hence we have that $\{\red{s}, \red{t} \} \subseteq \HHH$ or
        $\{\rdneg \red{s}, \rdneg \red{t} \} \subseteq \HHH$.

  \item $\red{\vec{\nabla}_{=}^\rightarrow}$: Let $\red{f} \rdeq^{\nu\tau} \red{g} \in \HHH$.
        Then, by definition, $f \beq^{\nu\tau} g \in \H$. By $\blue{\vec{\nabla}_{\mathfrak{f}}^+}$
        it follows that ${f}\;s\beq^{\nu} {g}\;s \in \H$ for every closed term
        $s_\tau \in \Lambda_\tau^c(\blue{\Sigma})$ .
        Hence we have that $\red{f}\;\red{s}\rdeq^{\nu} \red{g}\;\red{s} \in \HHH$ for every closed
        $\red{s}_\tau \in \Lambda_\tau^c(\red{\Sigma})$.

  \item $\red{\vec{\nabla}_{=}^r}$: Assume $\rdneg(\red{s} \rdeq^i \red{s}) \in \HHH$.
        Then, by definition, $\bneg(s \beq^i s) \in \H$, contradicting $\blue{\vec{\nabla}_{=}^r}$.
        Hence we have that $\rdneg(\red{s} \rdeq^i \red{s}) \notin \HHH$.

  \item $\red{\vec{\nabla}_{=}^u}$: Let $(\red{s} \rdeq^i \red{t}) \in \HHH$ and
        $\rdneg(\red{u} \rdeq^i \red{v}) \in \HHH$.
        Then, by definition, $(s \beq^i t) \in \H$ and $\bneg(u \beq^i v) \in \H$.
        By $\blue{\vec{\nabla}_{\mathit{m}}}$ it follows that
        $(s \beq^i t) \bneq^o (u \beq^i v) \in \H$. By $\blue{\vec{\nabla}_{\mathit{d}}}$
        we then have $s \bneq^i u \in \H$ or $t \bneq^i v \in \H$.
        Hence, $\rdneg(\red{s} \rdeq^i \red{u}) \in \HHH$ or $\rdneg(\red{t} \rdeq^i \red{v}) \in \HHH$.
        For the second half, we know by Corollary~\ref{corollary:equalitysaturated} that $\H$ is saturated
        and hence by Lemma~\ref{lemma:saturated}\ref{lemma:saturated3} it holds that $(t \beq^i s) \in \H$.
        Applying $\blue{\vec{\nabla}_{=}^s}$ yields
        $t \bneq^i u \in \H$ or $s \bneq^i v \in \H$, and thus
        $\rdneg(\red{t} \rdeq^i \red{u}) \in \HHH$ or $\rdneg(\red{s} \rdeq^i \red{v}) \in \HHH$.

\end{description}
\end{proof}

We claim that this result analogously holds if the original definitions of Andrews for the logical connectives
are assumed instead of the slightly modified ones introduced in \S\ref{sec:steen} and
used in~\cite{DBLP:phd/dnb/Steen18}; a technical proof remains future work.

%%%%%%%%%%%%%%%%%%%%%%%%%%%%%%%%%%%%%%%%%%%%%%%%%%%%%%%%%%%%%%%%%%%%%%%%%%%%%%%%%%%%%%%%
%%%%%%%%%%%%%%%%%%%%%%%%%%%%%%%%%%%%%%%%%%%%%%%%%%%%%%%%%%%%%%%%%%%%%%%%%%%%%%%%%%%%%%%%
%%%%%%%%%%%%%%%%%%%%%%%%%%%%%%%%%%%%%%%%%%%%%%%%%%%%%%%%%%%%%%%%%%%%%%%%%%%%%%%%%%%%%%%%
\section{Use Case: Bridging Model Existence}\label{sec:bridge}
\newcommand{\M}{\mathcal{M}}
\newcommand{\MM}{\mathcal{M}_{\blue{\triangleleft}}}

In this section, we apply the above reduction to derive a model existence theorem for 
Steen's Hintikka sets. Informally, we proceed as follows:
There exists an extensional $\{\req, \rneg\}$-model $\mathcal{M} \in \mathfrak{M}_{\beta\mathfrak{b}\mathfrak{f}}$
such that $\mathcal{M} \models \HHH$. Because we constructed $\HHH$ as an extensional Hintikka set
based on negation and equality, we know that the domain of Booleans in
$\mathcal{M}$ is bivalent.
In order to get a model solely based on equality (as required in the notion of Steen), we subsequently 
restrict $\mathcal{M}$ to terms over $\{ \req^\tau \mid \tau \in \types \}$.
Finally, we find an extensional model over frames isomorphic to it.\footnote{
We take a slight indirection here: We first assume negation is part of the red signature, to make sure there exists
an element in $\mathrm{n} \in \mathcal{D}_{oo}$ that is the interpretation of negation. Sadly, it seems there is currently
no easier way to enforce its existence; a more convenient way would be to show that there is an extensional model that satisfies
$\mathcal{L}_{\neg}(\mathrm{n})$ without having negation in the signature, cf.~\cite{brown07,J6} for details on these
semantic $\mathcal{L}$-properties.
}

First, we summarize important results by Brown~\cite{brown07} used in this reduction.

\begin{theorem}[Model Existence for Extensional Hintikka Sets~{\cite[Theorem 5.7.17]{brown07}}] \label{theorem:brown1}
Let $\H$ be an extensional $\Sigma$-Hintikka set
(i.e., $\H \in \mathfrak{Hint}_{\beta\mathfrak{b}\mathfrak{f}}(\Sigma)$).
There is an extensional $\Sigma$-model $\mathcal{M}$
(i.e., $\mathcal{M} \in \mathfrak{M}_{\beta\mathfrak{b}\mathfrak{f}}(\Sigma)$)
such that $\mathcal{M} \models \H$.
\end{theorem}

\begin{theorem}[Property $\mathfrak{b}$~{\cite[Theorem 3.3.7]{brown07}}] \label{theorem:brown2}
Let $\Sigma$ be a signature and $\mathcal{M}$ be an $\Sigma$-model.
Suppose either $\rtop, \rbot \in \Sigma$ or $\rneg \in \Sigma$.
Then $\mathcal{M}$ satisfies $\mathfrak{b}$ iff $\mathcal{D}_o$ has
two elements.
\end{theorem}

\begin{theorem}[Isomorphic Models over Frames~{\cite[Theorem 3.5.6]{brown07}}] \label{theorem:brown3}
Let $\mathcal{M} = \big( \mathcal{D}, @, \mathcal{E}, v \big) \in \mathfrak{M}_{\beta\mathfrak{b}\mathfrak{f}}(\red{\Sigma})$
be an extensional $\Sigma$-model such that $\mathcal{D}_o$ has
two elements. There is an 
isomorphic $\Sigma$-model $\mathcal{M}^h = \big(\mathcal{D}^h, @^h, \mathcal{E}^h, v^h \big)$ over frames, in particular
$\mathcal{D}^h_o = \{T,F\}$ and $v^h$ is the identity.
\end{theorem}

Now we infer a model existence theorem for Steen's Hintikka properties by bridging to those
of Brown:
Let $\H \in \mathfrak{H}$ be a Hintikka set (according to Steen) over a signature $\blue{\SigmaEq}$.
Let $\HHH$ be the translated set according to
\S\ref{sec:reduction}. By Theorem~\ref{thm:reduction} it follows that $\HHH$ is an
extensional $\red{\Sigma}$-Hintikka set, for $\red{\Sigma} := \{ \rneg \} \cup \{ \req^\tau \mid \tau \in \types \} \cup \{ \red{p} \mid p \in \blue{\SigmaEq} \text{ is parameter} \}$.
By Theorem~\ref{theorem:brown1} it follows that there is
an extensional model $\mathcal{M} = \big( \mathcal{D}, @, \mathcal{E}, v \big) \in \mathfrak{M}_{\beta\mathfrak{b}\mathfrak{f}}(\red{\Sigma})$ such
that $\mathcal{M} \models \HHH$.
Since $\rneg \in \red{\Sigma}$ it follows by Theorem~\ref{theorem:brown2} that $\mathcal{D}_o$ is bivalent.\footnote{
The fact that $\M$ satisfies property $\mathfrak{b}$ follows directly from the fact that
$\M \in \mathfrak{M}_{\beta\mathfrak{b}\mathfrak{f}}(\red{\Sigma})$.
}

Now we eliminate $\rneg$ from the signature $\red{\Sigma}$ to get an extensional model over 
$\{ \req^\tau \mid \tau \in \types \}$; we refer to this signature as $\red{\SigmaEq}$, i.e., 
let $\red{\SigmaEq} := \red{\Sigma} \setminus \{ \rneg \}$.
To that end, let $\mathcal{M}^\prime := \big( \mathcal{D}, @, \mathcal{E}|_{\Lambda(\red{\SigmaEq})}, v \big)$.
$\mathcal{M}^\prime$ is an extensional $\red{\SigmaEq}$-model, in
particular $\mathcal{D}_o$ is bivalent~\cite[Theorem 3.3.15]{brown07}.
Now by Theorem~\ref{theorem:brown3} it follows that there is an $\red{\SigmaEq}$-model over frames
$\mathcal{M}^h = (\mathcal{D}^h, @^h, \mathcal{E}^h, v^h)$ isomorphic to $\mathcal{M}^\prime$.

We now construct our desired Henkin model $\MM$ for $\H$ over $\blue{\SigmaEq}$ as follows:
%Let $\varphi$ be some variable assignment. 
%Then 
$\MM := (\mathcal{D}_{\blue{\triangleleft}}, \mathcal{I}_{\blue{\triangleleft}})$, where
\begin{itemize}
  \item $\mathcal{D}_{\blue{\triangleleft}} := \mathcal{D}^h$, and
  \item $\mathcal{I}_{\blue{\triangleleft}} := c \mapsto \mathcal{E}^h(\red{c})$ 
        for all $c \in \blue{\SigmaEq}$.
  %\item $\mathcal{I}_{\blue{\triangleleft}}$ such that, for all $\tau \in \types$,
  %  $\mathcal{I}_{\blue{\triangleleft}}(\beq^\tau) \equiv \mathcal{E}^h_\varphi(\req^\tau)$, and
  %  for every parameter $p \in \blue{\SigmaEq}$ let
  %  $\mathcal{I}_{\blue{\triangleleft}}(p) \equiv \mathcal{E}^h_\varphi(\red{p})$.
\end{itemize}

It is easy to see that $\MM$ is a Henkin model; in particular
$\mathcal{I}(\beq^\tau)(\mathsf{a},\mathsf{b}) = T$ iff $\mathsf{a} \equiv \mathsf{b}$ for every 
$\mathsf{a},\mathsf{b} \in \mathcal{D}^h_\tau$, $\tau \in \types$,
by property $\mathcal{L}_{=^\tau}(\mathsf{q})$ of $\mathcal{M}^h$ 
for $\mathsf{q} \equiv \mathcal{E}(\req^\tau) \in \mathcal{D}^h_{o\tau\tau}$.
Finally we need to verify that $\MM \models \H$ indeed holds.
Let $s_o \in \H$. By definition $\MM \models s$ if and only if
$\|s\|^{\MM ,g} \equiv T$ for every $g$. 
An induction over the structure of $s$ yields the desired result:
If $s$ is an equality of the form $(l \beq^\tau r)$, then
$\|l \beq^\tau r\|^{\MM ,g} \equiv T$ iff $\|l\|^{\MM ,g} \equiv \|r\|^{\MM ,g}$.
Since $\red{l} \req^\tau \red{r} \in \HHH$, we know $\M \models \red{l} \req^\tau \red{r}$
and, consequently, that $\mathcal{E}_\varphi(\red{l}) \equiv \mathcal{E}_\varphi(\red{r})$
for every assignment $\varphi$.
It follows that $\mathcal{E}^h_\varphi(l[\beq\backslash\req]) \equiv \mathcal{E}^h_\varphi(r[\beq\backslash\req])$ 
and hence, by definition of $\MM$ and the induction hypothesis, 
$\mathcal{I}_{\blue{\triangleleft}}(\beq^\tau)(\|r\|^{\MM ,g},\|r\|^{\MM ,g}) = T$ 
and thus $\|l \beq^\tau r\|^{\MM ,g} \equiv T$, for every $g$.
For parameters $p \in \blue{\SigmaEq}$, the proposition follows directly.
For complex formulas of the form $(s\;\overline{s^n})$ and for abstractions,
 we apply the induction hypotheses to
every sub-formula.
This construction yields:

\begin{theorem}[Bridged Model Existence]
Let $\H \in \mathfrak{H}$ be a $\blue{\SigmaEq}$-Hintikka set. Then there exists a $\blue{\SigmaEq}$-Henkin model $\mathcal{M}$
such that $\mathcal{M} \models \H$.
\end{theorem}

\subsection*{Acknowledgements.} We thank an anonymous reviewer for very
valuable feedback to this work (cf.~footnote~\ref{footnote:1}).

%%%%%%%%%%%%%%%%%%%%%%%%%%%%%%%%%%%%%%%%%%%%%%%%%%%%%%%%%%%%%%%%%%%%%%%%%%
%%%%%%%%%%%%%%%%%%%%%%%%%%%%%%%%%%%%%%%%%%%%%%%%%%%%%%%%%%%%%%%%%%%%%%%%%%
%%%%%%%%%%%%%%%%%%%%%%%%%%%%%%%%%%%%%%%%%%%%%%%%%%%%%%%%%%%%%%%%%%%%%%%%%%

\bibliographystyle{plain}
% \bibliography{proofs.bib}

\end{document}